\documentclass[journal]{IEEEtran}
\newif\ifIFAC                                                                  
\IFACfalse
                                                                               
\usepackage{graphicx}

\usepackage[cmex10]{amsmath}
\usepackage{amsfonts}
\usepackage{amssymb}
\usepackage{dsfont}
\usepackage{accents}
\usepackage{url}

\usepackage{empheq}
\usepackage[normalem]{ulem}

\usepackage{enumerate}
\usepackage{enumitem} 

\usepackage{tikz} 
\usepackage{pgfplots}
\usetikzlibrary{patterns}

\DeclareMathOperator{\diag}{diag}

\DeclareMathOperator*{\argmin}{arg\,min}

\newcommand{\ind}[1]{\text{I}\{#1\}}

\newcommand{\E}[1]{\mathbb{E}\big\{#1\big\}}

\usepackage{mathtools}
\DeclarePairedDelimiter{\PrfencesHard}{[}{]}
\DeclarePairedDelimiter{\PrfencesHardBig}{\Big[}{\Big]}
\DeclarePairedDelimiter{\PrfencesSoft}{(}{)}
\renewcommand{\Pr}{\operatorname{Pr}\PrfencesHard}
\newcommand{\PR}{\operatorname{Pr}\PrfencesHardBig}
\renewcommand{\diag}{\operatorname{diag}\PrfencesSoft}
\renewcommand{\vec}{\operatorname{vec}\PrfencesSoft}
\newcommand{\rank}{\operatorname{rank}\PrfencesSoft}

\newcommand{\fref}[1]{Fig.~\ref{#1}}

\newcommand{\eref}[1]{(\ref{#1})}
\newcommand{\sref}[1]{Section~\ref{#1}}

\newcommand{\thrmref}[1]{Theorem~\ref{#1}}

\newcommand{\lref}[1]{Lemma~\ref{#1}}

\usepackage{xcolor}
\newcommand{\comment}[1]{}  

\ifIFAC
\else
\usepackage{amsthm}
\fi
\newtheorem{theorem}{Theorem}
\newtheorem{lemma}{Lemma}
\newtheorem{problem}{Problem}
\newtheorem{assumption}{Assumption}

\usepackage{dsfont}
\newcommand{\ones}{\mathds{1}}

\tikzset{square arrow/.style={to path={-- ++(0,-.27) -| (\tikztotarget)}}}

\usepackage{subcaption}

\usepackage{cite} 
\usepackage{balance} 

\newcommand{\polylb}{\underaccent{\bar}{\Pi}}
\newcommand{\pii}{\pi_\infty}
\newcommand{\hMii}{\hat M_\infty}

\newcommand{\thetaNR}{\hat \theta_\text{NR}}
\newcommand{\thetaMM}{\hat \theta_\text{MM}}

\newcommand{\thetainit}{\hat \theta_\text{init}}
\newcommand{\hesslik}{\nabla_\theta^2 l_N}
\newcommand{\gradlik}{\nabla_\theta l_N}
\newcommand{\Fisher}{I_F}

\newcommand{\op}{o_\text{p}}
\newcommand{\Op}{\mathcal{O}_p}

\usepackage[english]{babel}
\usepackage{blindtext}

\allowdisplaybreaks

\ifCLASSINFOpdf
\else
\fi
\begin{document}
%
\title{Asymptotically Efficient Identification of Known-Sensor Hidden Markov Models}
%
%
%

\author{Robert~Mattila,
        Cristian~R.~Rojas,~\IEEEmembership{Member,~IEEE},\\
        Vikram~Krishnamurthy,~\IEEEmembership{Fellow,~IEEE}
        and~Bo~Wahlberg,~\IEEEmembership{Fellow,~IEEE}
\thanks{This work was partially supported by the Swedish Research
Council and the Linnaeus Center ACCESS at KTH. Robert Mattila, Cristian R. Rojas and Bo Wahlberg are with the Department of
    Automatic Control and ACCESS, School of Electrical Engineering, KTH Royal
    Institute of Technology. Stockholm, Sweden. (e-mails: 
\{rmattila,crro,bo\}@kth.se).
Vikram Krishnamurthy is with the Department of Electrical and Computer
Engineering, Cornell University. Cornell Tech, NY, USA. (e-mail:
vikramk@cornell.edu).}%
        }

        \markboth{}{}%
%



\maketitle

\begin{abstract}
    We consider estimating the transition probability matrix of a finite-state
    finite-observation alphabet hidden Markov model with known observation probabilities.
    The main contribution is a two-step algorithm; a method of moments estimator
    (formulated as a convex optimization problem) followed by a single iteration of
    a Newton-Raphson maximum likelihood estimator. The two-fold contribution of this
    letter is, firstly, to theoretically show that the proposed estimator is consistent
    and asymptotically efficient, and secondly, to numerically show that the method is
    computationally less demanding than conventional methods -- in particular for large
    data sets. 
\end{abstract}


%
\IEEEpeerreviewmaketitle

\section{Introduction}
%
%
%
%

\IEEEPARstart{T}{he} \emph{hidden Markov model} (HMM) has been applied in a diverse range
of fields, e.g., signal processing \cite{krishnamurthy_partially_2016}, gene sequencing
\cite{durbin_biological_1998, vidyasagar_hidden_2014} and speech recognition
\cite{rabiner_tutorial_1989}. The standard way of estimating the parameters of an HMM is
by employing a \emph{maximum likelihood} (ML) criterion. However, numerical ``hill
climbing'' algorithms for computing the ML estimate, such as direct maximization using
Newton-Raphson (and variants, e.g., \cite{cappe_inference_2005}) and the
\emph{expectation-maximization} (EM, e.g, \cite{dempster_maximum_1977,
rabiner_tutorial_1989}) algorithm are, in general, only guaranteed to converge to local
stationary points in the likelihood surface.  It is also known that these schemes can,
depending on the initial starting point of the algorithms, the shape of the likelihood
surface and the size of the data set, exhibit long run-times\comment{ (e.g.,
\cite{srebro_local_2007})}.

\comment{Even in the infinite-sample limit, EM, or other local-search methods,
    might take a very large number of steps to traverse these near-plateaus and
    converge. For this reason the conjecture does not directly imply
tractability.}

An alternative to ML criterion is to match moments of an HMM, resulting in a \emph{method
of moments} estimator (see, e.g. \cite{kay_fundamentals_1993} for details). In such
a method, observable correlations in the HMM data are related to the parameters of the
system. The correlations are empirically estimated and used in the inverted relations to
recover parameter estimates. A number of methods of moments for HMMs have been proposed in
the recent years, e.g., \cite{lakshminarayanan_non-negative_2010, anandkumar_method_2012,
    hsu_spectral_2012, kontorovich_learning_2013, anandkumar_tensor_2014,
mattila_recursive_2015, subakan_method_2015}. The main benefits over iterative ML schemes
are usually consistency and a shorter run-time, however, typically since only low-order
moments are considered, there is a loss of efficiency in the resulting estimate.

In the present letter, the problem of estimating the transition probabilities of
a finite discrete-time HMM with known sensor uncertainties, i.e., observation
matrix, is considered. This setup can be motivated in two ways: firstly, it can
be seen as the second step in a \emph{decoupling approach} to learning the HMM
parameters (see \cite{kontorovich_learning_2013}), or alternatively, by any
application where the sensor used to measure the system is designed/known to
the user. 

The main idea in this letter is a hybrid two-step algorithm based on combining the
advantages of the two aforementioned approaches. The first step uses a method of moments
estimator which requires a single pass over the data set (compared to iterative
algorithms, such as EM, that require multiple iterations over the data set). The second
step uses the method of moments estimate to initialize a non-iterative second-order direct
likelihood maximization procedure. This allows us to avoid resorting to ad hoc heuristics
for localizing a good starting point. More importantly, we show that it is
\emph{sufficient to perform only a single iteration} of the ML procedure to obtain an
asymptotically efficient estimate. Put differently, only two passes through the data set
are necessary in order to obtain an asymptotically efficient estimate.

To summarize, the main contributions of this letter are:
\begin{itemize}
\item a proposed two-step identification algorithm\comment{, for the system dynamics of 
    an HMM with known sensor uncertainties,} that exploits the benefits of
        both the method of moments approach (low computational burden and consistency) and
        direct likelihood maximization (high accuracy);
    \item we prove the consistency and asymptotic efficiency of the proposed
        estimator. Hence, the problem of only local convergence that may haunt
        iterative ML algorithms, such as EM, is shown to be
        avoided;
    \item numerical studies that show that the proposed method is up to an order of
        magnitude faster than the standard EM algorithm -- with the same resulting
        accuracy (when the EM iterations approach the global optimum of the likelihood
        function). Moreover, the run-time is, roughly, constant for a fixed data size,
        whereas the run-time of EM is highly dependent on the data (due to the number of
        iterations needed for convergence).
\end{itemize}

The outline of the remaining part of this letter is as follows. We first present a brief
overview of related work below.  \sref{sec:preliminaries} then poses the problem formally
and \sref{sec:algorithm} presents the algorithm. In \sref{sec:analysis} asymptotical
efficiency is proven, and \sref{sec:numerical_results} presents numerical studies.
\comment{\sref{sec:conclusion} concludes the paper with a brief summary.}

\subsection*{Related Work}
\label{sec:related_work}

HMM parameter estimation is now a classical area (with more than 50 years of literature).
There has recently been interest in the machine learning community for employing methods
of moments for HMMs.  The method presented in \cite[Appendix A]{hsu_spectral_2012}
demonstrates how to recover explicit estimates of the transition and observation matrices
by exploiting the special structure of the moments of an HMM. This method has been further
generalized and put in a tensor framework; see, e.g., \cite{anandkumar_method_2012},
\cite{anandkumar_tensor_2014} and references therein. The appealing attribute of these
methods is that they generate non-iterative estimates using simple linear algebra
operations (eigen and singular-value decompositions). However, the non-negativity and
sum-to-one properties of the estimated probabilities cannot be guaranteed.

There are a number of proposed methods of moments for HMMs formulated as
optimization problems (which allow constrains to be forced on the estimates),
e.g., \cite{lakshminarayanan_non-negative_2010},
\cite{kontorovich_learning_2013} and \cite{subakan_method_2015}. \comment{The
optimization formulation addresses the problem of possibly non-stochastic
estimates (i.e., estimates that do not fulfill the non-negativity and
sum-to-one constraints).} The identification problem is \emph{decoupled} in
\cite{kontorovich_learning_2013} into two stages: first an estimation of the
output parameters, and then a moment matching optimization problem. The
resulting optimization problem is related to the one in
\cite{lakshminarayanan_non-negative_2010} and to the problem in the present work. The
method we propose in this letter could be seen as a possible improvement of the
second step in the setting of \cite{kontorovich_learning_2013}.

In the general setting, hybrid approaches, such as the combination of EM and
direct likelihood maximization, and other attempts to accelerate EM has been
studied in, e.g., \cite{meilijson_fast_1989,fessler_space-alternating_1994}.
Iterative direct likelihood maximization for HMMs, as well as methods for
obtaining the necessary gradient and Hessian expressions, are treated in, e.g.,
\cite{lystig_exact_2002, cappe_recursive_2005, cappe_inference_2005,
turner_direct_2008, khreich_survey_2012, macdonald_numerical_2014}. The
combination of a method of moments and EM has, in the case of HMMs, been
considered in \cite{kontorovich_learning_2013}. 

\comment{
In \cite{lakshminarayanan_non-negative_2010}, non-negative matrix factorization is
employed to solve an optimization problem similar to ours. However, no
convergence results are given for the algorithms presented.
\cite{subakan_method_2015} deals with the special case of left-to-right and
Bakis HMMs as a two stage procedure where one stage involves a convex
optimization problem similar to the one in the present work.}

\comment{
    \cite{mattila_recursive_2015} presented how the algorithm in current work could
be employed in a recursive setting. There, convergence was not discussed and no
comparison to ML estimation was performed.
}

\section{Preliminaries and Problem Formulation}
\label{sec:preliminaries}

All vectors are column vectors unless transposed, $\mathds{1}$ denotes the vector of all
ones. The vector operator $\text{diag}:\mathbb{R}^n \rightarrow \mathbb{R}^{n\times n}$
gives the matrix where the vector has been put on the diagonal, and all other elements are
zero. $\| \cdot \|_F$ denotes the Frobenius norm of a matrix. The element at row $i$ and
column $j$ of a matrix is $[\cdot]_{ij}$, and the element at position $i$ of a vector is
$[\cdot]_i$.  Inequalities ($>, \geq, \leq, <$) between vectors or matrices should be
interpreted elementwise. The indicator function $\ind{\cdot}$ takes the value 1 if the
expression $\cdot$ is fulfilled and 0 otherwise. Let $\to_p$ and $\to_d$ denote
convergence in probability and in distribution, respectively, and let $\mathcal{O}_p$ and
$\op$ be stochastic-order symbols. $\sim$ denotes ``distributed according to''.

\subsection{Problem Formulation}

Consider a discrete-time finite-state \emph{hidden Markov model} (HMM) on the
state space $\mathcal{X} = \{1, 2, \dots, X\}$ with the transition probability
matrix
\begin{equation}
    [P]_{ij} = \Pr{x_{k+1} = j | x_k = i}.
\end{equation} 
Observations are made from the set $\mathcal{Y} = \{1, 2, \dots, Y\}$ according
to the observation probability matrix
\begin{equation}
    [B]_{ij} = \Pr{y_k = j | x_k = i}.
\end{equation}
These matrices are row-stochastic, i.e., the elements in each row sum to one.
Denote the initial distribution as $\pi_0$ and the stationary distribution as
$\pi_\infty$.

The HMM moments are joint probabilities of tuples of observations. The second
order moments can be represented by $Y \times Y$ matrices $M_k$ with elements
\begin{equation}
    [M_k]_{ij} = \Pr{y_{k} = i, y_{k+1} = j}.
    \label{eq:Mk_def}
\end{equation}
The following equation relates the second order moments and the system
parameters,
\begin{equation}
    M_k = B^T \diag{(P^T)^k \pi_0} PB,
    \label{eq:2nd_moments}
\end{equation}
and is the key to the method of moments formulation of the problem.

As we are interested in the asymptotic behaviour, we make the assumption that the initial
distribution $\pi_0$ is known to us -- its influence will anyway diminish over time. The
most important assumption we make is that the observation probabilities $B$ are known.
There are three motivations for this assumption: i) it admits the problem to a convex
formulation, ii) it holds in any real-world application where the sensor is designed by
the user, and iii) our method can be seen as an intermediate step of the \emph{decoupling}
approach in \cite{kontorovich_learning_2013}.  The identification problem we consider is,
hence,
\begin{problem}
    Consider an HMM with known initial distribution $\pi_0$ and known observation matrix
    $B$. The HMM is initialized according to $\pi_0$ and a sequence of observations $y_0,
    y_1, \dots, y_N$ is obtained. Given the sequence of $N+1$ observations
    $\{y_k\}_{k=0}^N$, estimate the transition matrix $P$.
    \label{pr:problem3}
\end{problem}

\comment{
The complete identification problem for HMMs is to estimate all the parameters
needed to specify the HMM, i.e., the transition matrix $P$, the observation
matrix $B$ and the initial distribution $\pi_0$, from a given set of
observations. In our case, as we consider the setup of a sequence of
consecutive observations, we are interested in the asymptotic behaviour of the
estimates and hence make the assumption that the initial distribution $\pi_0$
is known. We also make the assumption that the sensor uncertainties are known
to the user. There are three motivations for this assumption: i) is that it
admits the problem to a convex formulation, ii) is that it holds in any
real-world application where the sensor is designed by the user, and iii) is
that our method can be seen as an intermediate step of the \emph{decoupling}
approach in \cite{kontorovich_learning_2013}. Hence, the problem we consider
is:
}

\section{Asymptotically Efficient Two-Step Algorithm}
\label{sec:algorithm}

In this section, we outline the two-step algorithm which is the main contribution of this
letter.

\comment{
In this section, we outline the algorithm we propose. Recall that in the method
of moments, parameters of the system are related to the observable moments. The
moments are then empirically estimated and the relations are used in reverse to
yield estimates of the parameters. This can be formulated as an optimization
problem, where the objective is to minimize the mismatch between the observed
moments and the analytical expression \eref{eq:2nd_moments}.
}

\subsection*{Step 1. Initial Method of Moments Estimate}

In light of \eref{eq:Mk_def}, use the empirical moments estimate
\begin{equation}
    [\hat M_\infty]_{ij} = \frac{1}{N} \sum_{k=0}^{N-1} \ind{y_k = i,
    y_{k+1} = j},
    \label{eq:moment_estimator}
\end{equation}
for the (stationary) second order moments. \comment{Notice that we
consider only one realization, and hence cannot hope to estimate $M_k$ for all
$k$. However, as the processes reaches stationarity, we will be able to
estimate the stationary moments using the above estimator.}

In the moment matching optimization problem, we need to impose the constraint that the
transition matrix is a valid stochastic matrix, that is: the non-negativity and sum-to-one
properties of its rows. We require that the transition matrix of the HMM is ergodic
(aperiodic and irreducible). This implies, first of all, that $\pi_\infty$ is the right
eigenvector of $P^T$ and therefore satisfies the condition $\pi_\infty = P^T \pi_\infty$,
and secondly, that $\pi_\infty$ has strictly positive entries. We therefore, also, include
in the optimization problem a polytopic bound $\polylb$ on $\pi_\infty$ such that for
a vector $x \in \polylb \Rightarrow x > 0$.\footnote{This polyhedron can, for example, be
    obtained if it is possible to \emph{a priori} lower bound the elements of the
    transition matrix $P$ using another matrix $L$. In particular, this is possible since
    then the stationary distribution $\pi_\infty$ lies in a polyhedron $\polylb$ spanned
    by the normalized (i.e., non-negative and with elements that sum to one) columns of
    the matrix $(I-L^T)^{-1}$ -- see \cite{courtois_polyhedra_1985} for details.}

To summarize, estimating the transition matrix $P$ involves solving the optimization
problem (as the limit is taken in equation \eref{eq:2nd_moments} towards stationarity):
\begin{align}
    \min_{\pi_\infty, P} & \quad \|\hat M_\infty - B^T \diag{\pi_\infty} PB \|_F^2 \notag \\
    \text{s.t.} & \quad P \geq 0, \;\;\; \quad \pii \geq 0, \notag \\
                & \quad P \mathds{1} = \mathds{1}, \quad \mathds{1}^T \pi_\infty = 1, \notag \\
                & \quad \pi_\infty \in \polylb, \quad \pi_\infty = P^T \pi_\infty.
    \label{eq:nl_matching} 
\end{align}
This is, in general, a non-convex optimization problem. The lemma below shows
that convex optimization techniques can be used to solve the problem.

\begin{lemma}
    The solution of problem \eref{eq:nl_matching} is obtainable by solving the convex
    problem
    \begin{align}
        \min_A \quad & \|\hat M_\infty - B^T A B\|_F^2 \notag \\
        \text{s.t.} \quad & A \geq 0, \mathds{1}^T A \mathds{1} = 1, \notag \\
                          & A \mathds{1} \in \underaccent{\bar}{\Pi},
        A\mathds{1} = A^T \mathds{1},
        \label{eq:convex_mom}
    \end{align}
    and using \eref{eq:recover_pi} and \eref{eq:recover_P}, see below, to recover
    $\pi_\infty$ and $P$ from the variable $A$.
    \label{thrm:convex_equivalence}
\end{lemma}
\begin{proof}
In problem \eref{eq:convex_mom}, we identify the product $\diag{\pii} P$ in
problem \eref{eq:nl_matching} as a new parameter $A$, i.e.,
\begin{equation}
    A = \diag{\pii} P,
    \label{eq:definition_A}
\end{equation}
and optimize over its elements instead of over $\pii$ and $P$ jointly.  Notice
that it is possible to recover $\pii$ and $P$ from $A$ as follows: Firstly,
recover $\pii$ from
\begin{equation}
    A \mathds{1} = \diag{\pii} P \mathds{1} = \pii,
    \label{eq:recover_pi}
\end{equation}
employing the fact that $P \mathds{1} = \mathds{1}$.
Secondly, recover $P$ from
\begin{equation}
    \diag{\pii}^{-1} A = \diag{\pii}^{-1} \diag{\pii} P = P.
    \label{eq:recover_P}
\end{equation}
The lemma follows by noting that the cost functions in problems
\eref{eq:nl_matching} and \eref{eq:convex_mom} are the same, and then mapping
feasible solutions between the two problems.
\end{proof}

Solving problem \eref{eq:convex_mom} requires only a single pass over the data to obtain
$\hat M_\infty$, and then solving a data-size independent convex (quadratic) optimization
problem to compute an estimate of the transition matrix $P$. The trade-off compared to ML
estimation\comment{ or direct likelihood maximization,}, which requires multiple iterations
over the observation data set, is of course between estimation accuracy and computational
cost: the method of moments outlined above employs only the second order moments and will
hence have disregarded some of the information in the observed data. 

\subsection*{Step 2. Single Newton-Raphson Step}

We propose to exploit the trade-off by first obtaining an estimate of $P$ using
the convex method of moments \eref{eq:convex_mom}, and then taking \emph{a
single} Newton-Raphson step on the likelihood function to increase the accuracy
of the estimate.

The (log-)likelihood function of the observed data is
\begin{equation}
    l_{N}(\theta) = \log \Pr{ \, \{y_k\}_{k=0}^N \, | \, x_0 \sim \pi_0; \, \theta \,},
\end{equation}
where $\theta$ is a parametrization of the transition matrix $P$. Denote the estimate
resulting from the method of moments \eref{eq:convex_mom} as $\thetaMM$. Then a single
Newton-Raphson step is performed as follows:\footnote{We assume that parametrization handles the
constraints, if not, then the Newton-Raphson step can be formulated as a constrained
quadratic program.}
\begin{equation}
    \thetaNR = \thetaMM - \big[ \hesslik(\thetaMM)\big]^{-1} \; \gradlik(\thetaMM),
    \label{eq:newton_step}
\end{equation}
where the gradient $\nabla_\theta l_N(\hat \theta)$ and Hessian $\nabla_\theta^2 l_N(\hat
\theta)$ can be computed recursively -- see e.g., \cite{lystig_exact_2002,
cappe_recursive_2005, cappe_inference_2005, turner_direct_2008, khreich_survey_2012,
macdonald_numerical_2014}.

Compared to direct maximization of the likelihood function using the
Newton-Raphson method (see, e.g., \cite{cappe_inference_2005,
turner_direct_2008}), this procedure is non-iterative and hence, the gradient
and Hessian \emph{need only to be computed once}.

\comment{We will in the next section show that this
yields an asymptotically efficient estimator, and in the subsequent section
show by numerical simulations that the computational burden is lower than for
EM.}

\section{Analysis}
\label{sec:analysis}

In this section we analyze the properties of the proposed algorithm. First we state the
assumptions.

\begin{assumption}
    The transition matrix $P$ has positive elements. The observation matrix $B$
    is given, has full rank and is positive. There is a polytopic bound on
    $\pii$ such that all components of $\pii$ are strictly greater than zero.
\end{assumption}

The following lemma establishes (strong) consistency of the method of moments procedure.
\begin{lemma}
    The estimates of $P$ and $\pi_\infty$ obtained using \eref{eq:recover_pi}
    and \eref{eq:recover_P} from problem \eqref{eq:convex_mom} with the
    estimator $\hat M_\infty$ in \eref{eq:moment_estimator}, converge to their
    true values as the number of observations $N \rightarrow \infty$ with
    probability one.
    \label{thrm:convergence_P_pi}
\end{lemma}
\begin{proof}[Proof (outline)]
The lemma follows by showing
\begin{enumerate}
    \item that the estimate $\hat M_\infty$ converges to $M_\infty$ (using a 
        law of large numbers, \cite[Theorem 14.2.53]{cappe_inference_2005});
    \item that the solution $\hat A$ of the optimization problem converges to
        $A$ (follows by the fundamental theorem of statistical learning
        \cite[Lemma 1.1]{campi_system_2006} and the convexity of the cost
        function \cite[Theorem 10.8]{rockafellar_convex_1970});
    \item that the solution of the optimization problem $\hat A$ can be
        uniquely mapped to $P$ and $\pi_\infty$.
\end{enumerate}
Full details are available in the supplementary material.
\end{proof}

\begin{figure*}[ht!]
    \centering
\begin{subfigure}[]{0.49\linewidth}
    \centering
    \begin{tikzpicture}[xscale=1.05]
      \begin{axis}[
          width=\linewidth, 
          grid=major, 
          grid style={dashed,gray!30}, 
          xlabel=samples, 
          ylabel=RMSE,
          x post scale=0.9, 
          y post scale=0.51, 
          xmode=log,
          ymode=log,
          legend pos=north east,
          legend style={font=\fontsize{6}{5}\selectfont}, 
          cycle list name=black white,
        ]

        \addplot table[x=N,y=mom,col sep=comma]{data/benchmark_X5_Y5_median.csv}; 
        \addplot table[x=N,y=em,col sep=comma]{data/benchmark_X5_Y5_median.csv}; 
        \addplot table[x=N,y=em_mom,col sep=comma]{data/benchmark_X5_Y5_median.csv}; 
        \addplot table[x=N,y=em_true,col sep=comma]{data/benchmark_X5_Y5_median.csv}; 
        \addplot table[x=N,y=newton,col sep=comma]{data/benchmark_X5_Y5_median.csv}; 

        \legend{MM, EM, EM-MM, EM-True, 2S}
      \end{axis}
    \end{tikzpicture}
    \label{fig:accuracy}
\end{subfigure}
\begin{subfigure}[]{0.49\linewidth}
    \centering
    \vspace{-0.25cm}
    \begin{tikzpicture}[xscale=1.05]
      \begin{axis}[
          width=\linewidth, 
          grid=major, 
          grid style={dashed,gray!30}, 
          xlabel=samples, 
          ylabel={time [sec]},
          x post scale=0.9, 
          y post scale=0.51, 
          xmode=log,
          ymode=log,
          ymax=1e3, 
          legend pos=north west,
          legend style={font=\fontsize{6}{5}\selectfont}, 
          cycle list name=black white,
        ]

        \addplot table[x=N,y=mom_time,col sep=comma]{data/benchmark_X5_Y5_median.csv}; 
        \addplot table[x=N,y=em_time,col sep=comma]{data/benchmark_X5_Y5_median.csv}; 
        \addplot table[x=N,y=em_mom_time,col sep=comma]{data/benchmark_X5_Y5_median.csv}; 
        \addplot table[x=N,y=em_true_time,col sep=comma]{data/benchmark_X5_Y5_median.csv}; 
        \addplot table[x=N,y=newton_time,col sep=comma]{data/benchmark_X5_Y5_median.csv}; 

      \end{axis}
    \end{tikzpicture}
    \label{fig:runtime}
\end{subfigure}
    \caption{RMSE and run-time simulation data.}
    \label{fig:numerical}
\end{figure*}

Next, we provide the main theorem of the letter.
\begin{theorem}
    The estimate $\thetaNR$ obtained by the two-step algorithm
    \eref{eq:convex_mom}-\eref{eq:newton_step} is asymptotically efficient, i.e., as $N
    \to \infty$,
    \begin{equation}
        \sqrt{N}(\thetaNR - \theta^*) \to_d \mathcal{N}(0, \Fisher^{-1}(\theta^*)),
    \end{equation}
    where $\mathcal{N}$ is a normal distribution, $\theta^*$ corresponds to the
    true parameters and $\Fisher$ is the Fisher information matrix.
    \label{thrm:as_eff}
\end{theorem}
\begin{proof}[Proof (outline)]
    The theorem follows by showing that
    \begin{enumerate}
        \item the estimate $\hat M_\infty$ follows a central limit theorem
            \cite[Corollary 5]{jones_markov_2004}, and
            using this, concluding that $\hMii = M_\infty
            + \mathcal{O}_p(N^{-1/2})$ \cite[Appendix A]{pollard_convergence_1984};
        \item this order in probability can be propagated through the
            optimization problem \eref{eq:convex_mom} to obtain a similar order
            on $\hat P$ and $\hat \pi_\infty$ \cite[Theorem
            2.1]{daniel_stability_1973};
        \item verifying that certain regularity conditions hold to ensure
            that we have a central limit theorem for the gradient and
            a law of large numbers for the Hessian matrix of the log-likelihood
            function \cite[Theorems 12.5.5 and 12.5.6]{cappe_inference_2005};
        \item verifying by explicit computation that the single Newton-Raphson
            step yields an asymptotically efficient estimator.
    \end{enumerate}
    Again, full details are available in the supplementary material.

\end{proof}

\section{Numerical Evaluation}
\label{sec:numerical_results}

In this section, we evaluate the performance of the proposed two-step algorithm and
compare it to the standard EM algorithm for ML estimation. The EM implementation of Matlab
R2015a was employed (modified as to account for the fact that the observation matrix is
assumed known). The first step of the proposed algorithm, i.e., solving the convex
optimization problem \eref{eq:convex_mom}, was performed using the CVX package
\cite{grant_cvx_2014}. The second step, i.e., the single Newton-Raphson update
\eref{eq:newton_step}, can be implemented in (at least) two ways. The first is to
recursively compute the gradient and Hessian as explained in, e.g.,
\cite{lystig_exact_2002, cappe_recursive_2005, cappe_inference_2005, turner_direct_2008,
khreich_survey_2012, macdonald_numerical_2014}. The second, and the one we opted for, is
to use \emph{automatic differentiation} (AD, e.g., \cite{griewank_evaluating_2008}). We
interfaced Matlab to the ForwardDiff.jl-package in Julia \cite{revels_forward-mode_2016}
in our implementation. A small regularization term was added to the Hessian. Each
simulation was run on an Intel Xeon CPU at 3.1 GHz.

\comment{
    We used random systems with the transition and observation matrices generated
    as
    \begin{equation}
        P = \text{normalize\_rows}\big\{I + \frac{1}{X} \; \mathcal{U}[0,1]\big\}
    \end{equation}
    and
    \begin{equation}
    B = \text{normalize\_rows}\big\{\begin{bmatrix}I & 0 \end{bmatrix} + \frac{1}{Y} \; \mathcal{U}[0,1]\big\},
    \end{equation}
    where $I$ is the identity matrix, normalize\_rows makes sure that the row sum
    is one and $\mathcal{U}$ denotes (a matrix of) uniformly distributed numbers.
    For the sake of robustness, we used as an elementwise lower bound one tenth of
    the minimum element of the true stationary distribution of each system.
}

\begin{figure}[b!]
\centering
\includegraphics{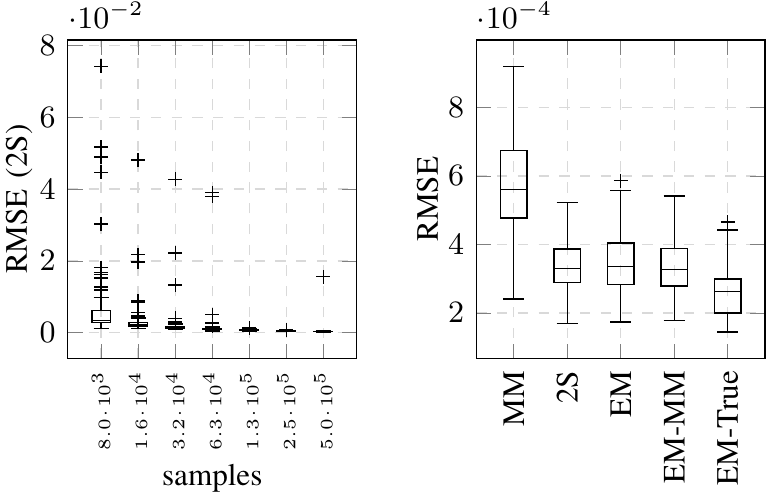}
\caption{Each box contains 100 simulations. \emph{(Left)} RMSE of the proposed algorithm
at different data sizes. \emph{(Right)} RMSE at $5\!\cdot\!10^5$ samples (one outlier not
seen).}
  \label{fig:box_plots}
\end{figure}

We sampled observations from randomly generated systems of size $X = Y = 5$.  Notice that
there are a total of 20 unknown parameters (i.e., elements of $P$) to estimate for such
systems.  We used an elementwise lower bound $\polylb$ of one tenth of the minimum element
of the true stationary distribution of each system. We compared the performance of the
proposed two-step algorithm (2S), to the estimate resulting from the method of moments
(MM), as well as, the EM algorithm started in three different initial points: a random
point (EM), the method of moments estimate (EM-MM) and the true parameter values
(EM-True). 

\fref{fig:numerical} presents the median over 100 simulations for each batch size of,
left, the root mean squared error (RMSE) and, right, the run-time. \comment{We stress the
logarithmic scale being used.}  \fref{fig:box_plots} presents box plots of, left, the
RMSEs of the proposed algorithm at various data sizes and, right, the RMSEs of the
compared algorithms for $5\cdot 10^5$ samples. All boxes contain 100 simulations. Three
things can be noted from the figures. 

Firstly, in the left plot of \fref{fig:numerical}, the loss of accuracy resulting from
only using the second order moments (compared to all moments in EM) is apparent from the
distance between the MM-curve and the EM-curves. This can also be seen in the right plot
of \fref{fig:box_plots}.

Secondly, also in the left plot of \fref{fig:numerical}, we see that the asymptotics
become valid around $10^5$ samples which takes the estimate resulting from the proposed
two-step method down to the accuracy of EM. The same conclusion is indicated by the left
plot of \fref{fig:box_plots}, where the number of observed outliers drop. These occurred
when the Hessian was not negative definite -- a result of the initial estimate not being
sufficiently close to the maximum of the likelihood function. Note that this can be
detected prior to employing the method.

Thirdly, in the right plot of \fref{fig:numerical}, it can be see that the run-times of
the compared algorithms differ by up to an order of magnitude. It should moreover be noted
that the run-time of the proposed algorithm is more or less constant for a fixed data size
(i.e., independent of the system and the observations), whereas the run-time of EM is
highly dependent on the data (due to the number of iterations needed to converge): The
maximum run-times for $5 \cdot 10^5$ observations were 1083, 480, 166 seconds for EM,
EM-MM and EM-True, respectively, whereas for the proposed method it was 54 seconds.

\comment{
We compare the algorithm to the Matlab implementation of the EM algorithm for
HMMs. Since EM is a local search algorithm, the choice of starting point is
crucial, in two senses: we are only guaranteed to converge locally, and
depending on the shape of the likelihood surface, convergence might be slow
(i.e. require many iterations). Also, each iteration of EM is computationally heavy
when the batch size of observations is large. The method of moments estimate is
obtained using only one pass over the data and then by solving a convex
optimization problem whose complexity is independent of the number of samples
(it depends only on the dimension of the system). 
}


\section{Conclusion}
\label{sec:conclusion}

This letter has proposed and analyzed a two-step algorithm for identification
of HMMs with known sensor uncertainties. A method of moments was combined with
direct likelihood maximization to exploit the benefits of both approaches:
lower computational cost and consistency in the former, and accuracy in the
later. Theoretical guarantees were given for asymptotic efficiency and
numerical simulations showed that the algorithm can yield the same accuracy as
the standard EM algorithm, but in up to an order of magnitude less time.

\ifCLASSOPTIONcaptionsoff
  \newpage
\fi




\clearpage

\balance

\bibliographystyle{IEEEtran}
\bibliography{IEEEabrv,../rob_references}

\clearpage
\nobalance

\clearpage

\appendices

Here, we provide details of the proofs stated in the paper.

\section{Proof of \lref{thrm:convex_equivalence}}

To show the equivalence, we first establish that the mappings between $P$ and
$\pii$, and $A$ are one-to-one using the following lemma.
\begin{lemma}
    The mappings \eref{eq:definition_A}, \eref{eq:recover_pi} and
    \eref{eq:recover_P} between $P$ and $\pi_\infty$, and $A$ are one-to-one.
    \label{lemma:unique_mapping}
\end{lemma}
\begin{proof}
    Recall that $\pii$ has strictly positive elements.  Firstly, given $P$ and
    $\pii$, the equation \eref{eq:definition_A} yields a single $A$. Secondly,
    assume $\diag{\pi_\infty} P = A = \diag{\tilde \pi_\infty} \tilde P$, where
    $\tilde P$ and $P$ are row-stochastic.  Multiplying the equation by $\ones$
    from the right yields $\pi_\infty = \tilde \pi_\infty$. Then, multiplying
    from the left by $\diag{\pi_\infty}^{-1}$ yields $P = \tilde P$.

    Thirdly, equations \eref{eq:recover_pi} and \eref{eq:recover_P} yield
    unique $P$ and $\pii$ given an $A$. Fourthly, assume $A$ and $\tilde A$
    both map to $P$ and $\pii$, i.e., $A\ones = \pii = \tilde A \ones$, and
    $\diag{A\ones}^{-1} A = P = \diag{\tilde A\ones}^{-1} \tilde A$.
    Multiplying the last equation by $\diag{A\ones}$ yields $A = \tilde A$.
\end{proof}
Then we note that the cost functions are the same in the two formulations
\eref{eq:nl_matching} and \eref{eq:convex_mom} of the problem.  Secondly, we
map feasible solutions between the two problems.
\subsubsection{Solution of \eref{eq:nl_matching} $\Rightarrow$ Solution of
\eref{eq:convex_mom}}

Assume $P$ and $\pii$ are optimal for \eref{eq:nl_matching}, and define $A
= \diag{\pii} P$. Then
\begin{itemize}
    \item $P \geq 0, \pii \geq 0 \Rightarrow A \geq 0$,
    \item $\ones^T A \ones = \ones^T \diag{\pii} P \ones = \pii^T \ones = 1$,
    \item $A\ones = \diag{\pii} P \ones = \diag{\pii} \ones = \pii \in \polylb$,
    \item $\pii = P^T \pii \Rightarrow \diag{\pii} \ones = P^T \diag{\pii} \ones
        \overset{P\ones = \ones}{\Rightarrow} \diag{\pii} P \ones
        = [\diag{\pii}P]^T \ones \Rightarrow A \ones = A^T \ones$.
\end{itemize}

\subsubsection{Solution of \eref{eq:convex_mom} $\Rightarrow$ Solution of
\eref{eq:nl_matching}}

Assume $A$ is optimal for \eref{eq:convex_mom}. Let $\pii = A \ones$ and $P =
\diag{A\ones}^{-1} A$. Note that $\diag{A \ones}^{-1}$ is well-defined since $A
\ones \in \polylb$, i.e., $A \ones > 0$. Then
\begin{itemize}
    \item $A \geq 0 \Rightarrow \pii \geq 0 \text{ and } P \geq 0$,
    \item Since for any vector $x$ with all non-zero elements it holds that $[\diag{x}^{-1} x]_i
        = [\diag{x}^{-1}]_{ii} [x]_i = [x]_i^{-1} [x]_i = 1$, i.e.,
        $\diag{x}^{-1} x = \ones$, we have that $\ones = \diag{A\ones}^{-1}
        A \ones
        = P \ones$,
    \item $\ones^T A \ones = 1 \Rightarrow \ones^T \pii = 1$,
    \item $\pii = A \ones \in \polylb$,
    \item $A \ones = A^T \ones \Rightarrow \pii = A^T \ones$. Then, again
        employing that $\diag{x}^{-1} x = 1$ for any vector with all non-zero
        elements; $\pii = A^T \ones = A^T \diag{A\ones}^{-1} A \ones
        = [\diag{A\ones}^{-1} A]^T A \ones = P^T \pii$.
\end{itemize}


\section{Proof of \lref{thrm:convergence_P_pi}}

Here, we provide a sequence of lemmas that give the details on that 
\begin{enumerate}
    \item[A.] the estimate $\hat M_\infty$ converges to $M_\infty$,
    \item[B.] the solution $\hat A$ of the optimization problem
        \eref{eq:convex_mom} converges to the true parameter $A$,
    \item[C.] the solution $\hat A$ of the optimization problem
        \eref{eq:convex_mom} can be converted
        uniquely to $\hat P$ and $\hat \pi_\infty$.
\end{enumerate}

\subsection{Convergence of $\hat M_\infty$}

\begin{lemma}
    Let the sequence of $N$ observations from the HMM with initial
    distribution $\pi_0$ be the set $\{y_k\}_{k=0}^N$ and form the empirical
    estimate \eref{eq:moment_estimator} of the second order moments.  Then the
    estimate converges, i.e.,
    \begin{equation}
        \hat M_\infty \rightarrow B^T \diag{\pi_\infty} P B
    \end{equation}
    with probability one as the number of observations $N \rightarrow \infty$.
    \label{lemma:convergence_M}
\end{lemma}
Before we prove the above lemma, we first introduce two auxiliary lemmas.
\begin{lemma}
    Let $x_k$ be the state of an HMM and $y_k$ the corresponding observation.
    Then the process defined by the tuplet $(x_k, y_k, y_{k-1})$ is a Markov
    chain. 
    \label{lemma:hmm_as_mc}
\end{lemma}
\begin{proof}
    It can be checked that the Markov property is satisfied.
\end{proof}
The above lemma allows us to recast the HMM into a Markov chain so that we can
leverage convergence results related to Markov chains. The following lemma
guarantees necessary properties of this new Markov chain:
\begin{lemma}
    If the transition and observation matrices of the HMM referred to in
    Lemma~\ref{lemma:hmm_as_mc} have all elements strictly positive, then the
    Markov chain $(x_k, y_k, y_{k-1})$ is irreducible and aperiodic.
    \label{lemma:ergodic_hmm_as_mc}
\end{lemma}
\begin{proof}
    The transition matrix of the lumped Markov chain consists of
    multiplications between elements of the $P$ and $B$ matrices (which are
    strictly positive) and zeros (whenever the common observation is not
    shared).
    
    It can be shown that any state can reach any state with positive
    probability in at most two steps ($\Rightarrow$ irreducibility).
    Furthermore, it can be shown that any state with the same two observations
    has positive probability of returning to itself in one step ($\Rightarrow$
    aperiodicity).
\end{proof}

With these two lemmas, we are now ready to prove \lref{lemma:convergence_M}.

\begin{proof}[Proof of \lref{lemma:convergence_M}]
    Denote the state of the lumped Markov chain as $z_k = (y_k, y_{k+1},
    x_{k+1})$ and let the functions $f_{ij}$ be given by
    \begin{equation}
        f_{ij}(z_k) = \ind{y_k = i, y_{k+1} = j}.
    \end{equation}
    Then,
    \begin{align*}
        \lim_{N\rightarrow\infty} [\hMii]_{ij} &= \lim_{N\to \infty} \frac{1}{N} \sum_{k=0}^{N-1}
        \ind{y_k = i, y_{k+1} = j} \\
        &= \lim_{N\to \infty} \frac{1}{N} \sum_{k=0}^{N-1} f_{ij}(z_k) \\
        &\rightarrow \sum_{z \in \mathcal{Z}} \tilde \pi_\infty(z) f_{ij}(z) \\
        &= \sum_{x\in\mathcal{X}}\tilde \pi_\infty((i,j,x)) \\
        &= \sum_{x\in\mathcal{X}} \lim_{k\to\infty} \Pr{(y_k, y_{k+1}, x_{k+1})
        = (i,j,x)} \\
        &= \lim_{k\to\infty} \Pr{(y_k, y_{k+1}) = (i,j)} \\
        &= [B^T \diag{\pi_\infty} PB]_{ij} \\
        &= [M_\infty]_{ij},
    \end{align*}
    with probability one, where $\mathcal{Z} = \mathcal{Y} \times \mathcal{Y}
    \times \mathcal{X}$ and $\tilde \pi_\infty$ is the stationary distribution of
    $z_k$. Here, we used that $z_k$ is aperiodic and irreducible
    (\lref{lemma:ergodic_hmm_as_mc}) and the strong law of large numbers for
    Markov chains \cite[Theorem 14.2.53]{cappe_inference_2005}.
\end{proof}

\subsection{Convergence of optimization solution}

As guaranteed by \lref{lemma:convergence_M}, the estimate $\hat M_\infty$ will
converge to $M_\infty$ with probability one. The next step is to
show that the solution(s) of the optimization problem \eref{eq:convex_mom}
converges to the optimal value, i.e., that $\hat A \rightarrow A$, as $\hat
M_\infty \rightarrow M_\infty$. The following lemma guarantees that there is
a unique solution to the optimization problem and allows us to write \emph{the}
minimizer from here on.
\begin{lemma}
    Under the assumption that $B$ has full rank, the minimizer of the
    optimization problem \eref{eq:convex_mom} is unique. 
    \label{lemma:unique_minimizer}
\end{lemma}
\begin{proof}
    We check that the Hessian of the cost function in \eref{eq:convex_mom} is
    positive definite to ensure strict convexity (see, e.g.,
    \cite{boyd_convex_2004}). Note that the cost is of the form (see equation
    \eref{eq:vectorization_of_cost} below for details)
    \begin{align}
        f(x) &= \| Qx + q \|_2^2 \notag \\
             &= x^T Q^T Q x + 2 q^T Q x + q^T q,
            \label{eq:quad_cost}
    \end{align}
    where $Q = B^T \otimes B^T$, which has the Hessian
    \begin{equation}
        \nabla^2 f(x) = 2 Q^T Q.
    \end{equation}
    Positive definiteness of the Hessian is, in this case, equivalent to
    \begin{align}
        x^T \big[\nabla^2 f(x)\big] x > 0 &\quad\forall x \neq
        0 & \Leftrightarrow \notag \\
        x^T \big[2Q^T Q] x > 0 &\quad\forall x \neq 0 & \Leftrightarrow \notag \\
        2 (Qx)^T (Q x) > 0 &\quad\forall x \neq 0 & \Leftrightarrow \notag \\
        2 \| Qx \|^2 > 0 &\quad\forall x \neq 0 & \Leftrightarrow \notag \\
        Qx \neq 0 &\quad\forall x \neq 0 & \Leftrightarrow \notag \\
        \ker Q  = \{0&\}. &
        \label{eq:strict_convex}
    \end{align}
    Since (see, e.g., \cite{horn_topics_1991}) 
    \begin{equation}
        \rank{B^T \otimes B^T} = \rank{B^T}\rank{B^T} = X^2,
    \end{equation}
    we see that the $Y^2 \times X^2$ matrix $B^T \otimes B^T$ has full column
    rank. By \eref{eq:strict_convex}, this implies uniqueness since, then, the
    cost function is strictly convex.
\end{proof}

The next lemma says that the sequence of minimizers of the approximate
optimization problems will converge to the minimizer of the true optimization
problem.
\begin{lemma}
    Let $\hat A$ be the minimizer of the optimization problem
    \eref{eq:convex_mom} using $\hat M_\infty$ from equation
    \eref{eq:moment_estimator}, and let $A$ be the minimizer of the optimization problem
    \eref{eq:convex_mom} using instead the true $M_\infty$.
    Then $\hat A$ converges with probability one to $A$ as $\hat M_\infty$ tends to
    $M_\infty$ as in \lref{lemma:convergence_M}.
    \label{lemma:epigraphical_convergence}
\end{lemma}
To be able to prove Lemma~\ref{lemma:epigraphical_convergence}, we will make
use of another two additional lemmas. The first one provides results regarding
how well a minimizer of an approximate cost function is with respect to the
minimizer of the true cost function. In summary, it says that if we have
uniform convergence of the cost function, the minimizer of the approximate cost
function will converge to the minimizer of the true cost function, if the
parameter set is compact.
\begin{lemma}
    Consider a family of random functions $f_k(\theta) : \Theta \to \mathbb{R}$, where
    $\Theta$ is a compact subset of some Euclidean space. Let $\theta_k
    = \argmin_{\theta \in \Theta} f_k(\theta)$. If $f_k(\theta)$ tends
    uniformly to a continuous (on $\Theta$) deterministic limit $\bar
    f(\theta)$ with probability one, i.e.,
    \begin{equation}
        \sup_{\theta \in \Theta} | f_k(\theta) - \bar f(\theta) | \to 0
        \label{eq:limit_assumption}
    \end{equation}
    with probability one as $k \to \infty$, then with $\Theta^* = \{ \theta \in \Theta \text{ s.t. } \theta
            \text{ minimizes } \bar f(\theta) \}$, we have that
            \begin{equation}
                \inf_{\theta \in \Theta^*} \| \theta_k - \theta \| \to 0
            \end{equation}
            with probability one as $k \to \infty$.
    \label{lemma:convergence_minimizers}
\end{lemma}
\begin{proof}
    This follows from Lemma 1.1 of \cite{campi_system_2006} by restricting to
    the realizations where \eref{eq:limit_assumption} hold. Similar results also appear in
    \cite{pollard_asymptotics_1993} and \cite[Ch. 24]{gourieroux_statistics_1995}.
\end{proof}
The second auxiliary lemma needed states roughly that if we have pointwise
convergence with probability one of a sequence of convex functions, then we
also have uniform convergence over compact sets with probability one.
\begin{lemma}
    Suppose $f_k(\theta)$ is a sequence of convex random functions defined on
    an open convex set $\mathcal{S}$ of some Euclidean space, which converges
    pointwise in $\theta$ with probability one to some $\bar f(\theta)$.
    Then
    \begin{equation}
        \sup_{\theta \in \Theta} | f_k(\theta) - \bar f(\theta) |
        \label{eq:sup_fk_f}
    \end{equation}
    tends to zero with probability one as $k \to \infty$, for each compact subset
    $\Theta$ of $\mathcal{S}$.
    \label{lemma:pointwise_to_uniform}
\end{lemma}
\begin{proof}
    This lemma follows from \cite[Theorem 10.8]{rockafellar_convex_1970} by restricting to
    the set of realizations where pointwise convergence holds, which has probability one
    by assumption.
\end{proof}
Combining these two lemmas allows us to provide proof for
Lemma~\ref{lemma:epigraphical_convergence}.
\begin{proof}[Proof of Lemma \ref{lemma:epigraphical_convergence}]
    The cost function in problem \eref{eq:convex_mom} is strictly convex (see
    proof of Lemma~\ref{lemma:unique_minimizer}). The set
    of feasible parameters is compact and convex. From
    \lref{lemma:convergence_M}, we know that $\hat M_\infty$ converges with
    probability one. Since the cost function is a continuous mapping of $\hat
    M_\infty$, we conclude that the cost function converges pointwise with
    probability one. Hence, the conditions of
    Lemma~\ref{lemma:pointwise_to_uniform} are fulfilled.  This in turn
    fulfills the conditions of Lemma~\ref{lemma:convergence_minimizers} which
    allows us to conclude that $\hat A$ will tend to $A$ with probability one.
\end{proof}

\subsection{Convergence of $\hat P$ and $\hat \pi_0$}

From Lemma~\ref{lemma:epigraphical_convergence}, we know that $\hat
A \rightarrow A$ as the number of samples tends to infinity. Since the mapping
from $A$ to $P$ and $\pii$ is unique (see \lref{lemma:unique_mapping}), we
conclude that the estimates of $P$ and $\pii$ also will converge. In summary,
this concludes the proof of \thrmref{thrm:convergence_P_pi}.

\section{Proof of \thrmref{thrm:as_eff}}

Parts of the proof are inspired by \cite{kontorovich_learning_2013}.

\subsubsection{Central limit theorem for $\hMii$}

We will show that a central limit theorem holds for the estimates. For this, we
employ the following theorem from \cite{jones_markov_2004}:
\begin{theorem}[Corollary 5 of \cite{jones_markov_2004}]
    Consider a uniformly ergodic Markov chain on $\mathcal{X}$ with stationary
    distribution $\pi_\infty$. Suppose $\mathbb{E}_{\pi_\infty} f^2(x)
    < \infty$, where $f:\mathcal{X}\to\mathbb{R}$. Then for any initial distribution, as $N \to \infty$,
    \begin{equation}
        \sqrt{N}(\bar f_N - \mathbb{E}_{\pi_\infty} f) \to_d \mathcal{N}(0,
        \sigma^2_f),
    \end{equation}
    where $\bar f_N = N^{-1} \sum_{k=1}^N f(x_k)$ and $\sigma_f^2 < \infty$ is
    a constant.
    \label{thrm:general_clt}
\end{theorem}

As in the proof of \lref{lemma:convergence_M}, denote the state of the lumped
Markov chain as $z_k = (y_k, y_{k+1}, x_{k+1}) \in \mathcal{Z}
= \mathcal{Y}\times\mathcal{Y}\times\mathcal{X}$ and let the functions $f_{ij}$
be given by
\begin{equation}
    f_{ij}(z_k) = \ind{y_k = i, y_{k+1} = j}.
\end{equation}
Then, as guaranteed by \thrmref{thrm:general_clt} ($z_k$ is uniformly ergodic
since it is finite -- see \cite[Example 1]{jones_markov_2004}),
\begin{equation}
    \sqrt{N} \Bigg( \frac{1}{N} \sum_{k=0}^{N-1} f_{ij}(z_k)
    - \sum_{z\in\mathcal{Z}} \tilde \pi_\infty(z) f_{ij}(z) \Bigg) \to_d
    \mathcal{N}(0, \sigma_{ij}^2),
\end{equation}
or by changing back to the original variables,
\begin{align}
    \sqrt{N} \Bigg( \frac{1}{N} &\sum_{k=0}^{N-1} \ind{y_k = i, y_{k+1} = j}
    \notag \\
    &- \lim_{k\to\infty} \Pr{y_k = i, y_{k+1} = j} \Bigg) \notag \\
    &\to_d \mathcal{N}(0, \sigma_{ij}^2),
\end{align}
i.e.,
\begin{equation}
    \sqrt{N} \Big( [\hMii]_{ij} - [M_\infty]_{ij} \Big) \to_d \mathcal{N}(0, \sigma_{ij}^2),
    \label{eq:hMii_convergence_in_p}
\end{equation}
where $\sigma_{ij}^2 < \infty$ are constants.

\subsubsection{$\sqrt{N}$-consistency of $\hat M_\infty$} We now establish that
$\hMii$ is a $\sqrt{N}$-consistent estimator using the above result and the
following lemma.
\begin{lemma}[Appendix A, \cite{pollard_convergence_1984}]
    If a sequence of random variables $Z_N$ and a constant $z_0$ tend in
    distribution to another random variable $Z$ (as $N \to \infty$) according
    to
    \begin{equation}
        \sqrt{N}(Z_N - z_0) \to_d Z,
    \end{equation}
    then
    \begin{equation}
        Z_n - z_0 = \Op(N^{-1/2}).
    \end{equation}
\end{lemma}
Leveraging the above lemma, we conclude that
\begin{equation}
    [\hMii]_{ij} = [M_\infty]_{ij} + \mathcal{O}_p(N^{-1/2}).
    \label{eq:ordo_p_theta}
\end{equation}
This, by definition, means that for every $\varepsilon > 0$, we can find
a constant $c_{ij}(\varepsilon)$ such that for all $N$ sufficiently large,
\begin{equation}
    \PR{\sqrt{N} \big| [\hMii]_{ij} - [M_\infty]_{ij} \big| > c_{ij}(\varepsilon)} < \varepsilon.
\end{equation}

\comment{
\subsubsection{(Original) $\sqrt{N}$-consistency of $\hat M_\infty$} We now establish that
$\hMii$ is a $\sqrt{N}$-consistent estimator using the above result.

\begin{theorem}[Slutsky's theorem]
    If a sequence of random variables $X_n$ tends in distribution to another
    random variable $X$, i.e., $X_n \to_d X$, then $X_n = \mathcal{O}_p(1)$.
    \label{thrm:slutsky}
\end{theorem}

\comment{
In our setting, we employ \thrmref{thrm:slutsky} for an estimator $\hat \theta$
that follows a central limit theorem,
\begin{equation}
    \sqrt{N}(\hat \theta - \theta_0) \to_d \mathcal{N}(0, \sigma^2),
\end{equation}
i.e., $\hMii$ as in equation \eref{eq:hMii_convergence_in_p}, to conclude that
\begin{equation}
    \hat \theta = \theta_0 + \mathcal{O}_p(N^{-1/2}).
    \label{eq:ordo_p_theta}
\end{equation}
}
We employ \thrmref{thrm:slutsky} in conjunction with equation
\eref{eq:hMii_convergence_in_p} to conclude that 
\begin{equation}
    \sqrt{N} \Big( [\hMii]_{ij} - [M_\infty]_{ij} \Big) = \mathcal{O}_p(1),
\end{equation}
or equivalently that
\begin{equation}
    [\hMii]_{ij} = [M_\infty]_{ij} + \mathcal{O}_p(N^{-1/2}).
    \label{eq:ordo_p_theta}
\end{equation}
This, by definition, means that for every $\varepsilon > 0$, we can find
a constant $c_{ij}(\varepsilon)$ such that for $N$ sufficiently large,
\begin{equation}
    \PR{\sqrt{N} \big| [\hMii]_{ij} - [M_\infty]_{ij} \big| > c_{ij}(\varepsilon)} < \varepsilon.
\end{equation}
}

\subsubsection{$\sqrt{N}$-consistency of $\hat A$}
\label{ssec:sqrtN_consistencyA}
We now propagate the
$\sqrt{N}$-consistency of $\hMii$ to the variable $\hat A$ through the
optimization problem \eref{eq:convex_mom}.

First note that problem \eref{eq:convex_mom} can be rewritten on the standard
form for a \emph{quadratic program} (QP),
\begin{align}
    \min_x \quad & \frac{1}{2} x^T Q x - q^T x \notag \\
    \text{s.t.} \quad & Gx \leq g,  \notag \\
                      & Dx = d,
    \label{eq:nominal_QP}
\end{align}
where $Q$ is a positive definite matrix. In particular, using the
identity (for arbitrary matrices $A$, $B$ and $C$ of appropriate dimensions, see, e.g.,
\cite{horn_topics_1991}) 
\begin{equation}
    \vec{ABC} = (C^T \otimes A) \vec{B},
\end{equation}
we have that 
\begin{align*}
    \| \hat M_\infty &- B^T A B \|_F^2 = \| \vec{\hMii - B^T A B} \|_2^2 \\
                                      &= \| \vec{\hMii} - \vec{B^T AB} \|_2^2 \\
                                      &= \| \vec{\hMii} - (B \otimes B)^T \vec{A} \|_2^2 \\
                                      &= \vec{\hMii}^T \vec{\hMii} \\
                                      &- 2 \vec{\hMii}^T (B \otimes B)^T \vec{A} \\
                                      &+ \vec{A}^T (B \otimes B) (B
    \otimes B)^T \vec{A}, \stepcounter{equation}\tag{\theequation}
    \label{eq:vectorization_of_cost}
\end{align*}
so that,
\begin{align}
    Q &= 2 (B \otimes B) (B \otimes B)^T, \\
    \hat q &= 2 (B \otimes B) \vec{\hMii},
\end{align}
where $x = \vec{A}$, and the constant term has been dropped. The constraints can
similarly be translated by vectorization, e.g., $\ones^T A \ones = 1$
translates to $\ones^T \vec{A} = 1$, i.e., $\ones^T x = 1$.

The uncertainty in this problem, resulting from the estimation procedure, lies
in the estimate of the moments $\hMii$. Note that this only influences the cost
function -- not the constraints. We now ask ourselves how the uncertainty in
$\hMii$ propagates through the QP into our variable of interest $\hat A$.

Denote the minimizer of the nominal problem \eref{eq:nominal_QP}, where
$M_\infty$ is used instead of $\hMii$, as $x^*$ (= $\vec{A}$) and let
\begin{equation}
    \hat x^* = \argmin_x \frac{1}{2} x^T Q x - \hat q^T x,
\end{equation}
subject to the same constraints as in problem \eref{eq:nominal_QP}. Then
\cite[Theorem 2.1]{daniel_stability_1973} provides the following bound on the distance
between the solution of the nominal QP and the solution of the perturbed QP:
\begin{equation}
    \|x^* - \hat x^*\|_2 \leq \frac{\delta}{\lambda - \delta} (1
    + \|x^*\|_2),
    \label{eq:qp_bound}
\end{equation}
where $\delta = \|q - \hat q\|_2$ and $\lambda$ is the smallest eigenvalue
of $Q$.\footnote{$\lambda > \delta$ holds if $N$ is large
    enough (so that $\delta$ is small enough).}

Let $\sigma_1(\cdot)$ denote the largest singular value, then we note that (for
every $\varepsilon > 0$)
\begin{align*}
    \delta &= \|q-\hat q\|_2 \notag \\
                   &= \| 2(B\otimes B) \vec{M_\infty} - 2(B\otimes B) \vec{\hMii} \|_2 \\
                   &= 2\, \| (B\otimes B) (\vec{M_\infty} - \vec{\hMii}) \|_2 \\
                   &\leq 2\,\sigma_1(B \otimes B) \; \|\vec{M_\infty} - \vec{\hMii}\|_2 \\
                   &\leq 2\,\sigma_1(B \otimes B) \; \|\vec{M_\infty} - \vec{\hMii}\|_1 \\
                   &= 2\,\sigma_1(B \otimes B) \; \sum_{i,j\in\mathcal{Y}}\big| [M_\infty]_{ij} - [\hMii]_{ij}\big| \\
                   &\leq 2\,\sigma_1(B \otimes B) \; \sum_{i,j\in\mathcal{Y}} \frac{c_{ij}(\varepsilon)}{\sqrt{N}} \\
                   &\leq \frac{1}{\sqrt{N}} \; 2 \, \sigma_1(B \otimes B) \; Y^2 \max_{i,j} c_{ij}(\varepsilon) \\
                   &\overset{\text{def.}}{=} \frac{1}{\sqrt{N}} \; K(\varepsilon),
                    \stepcounter{equation}\tag{\theequation}
\end{align*}
with probability greater than $1-\varepsilon$, where $c_{ij}(\varepsilon)$ are
the constants in the stochastic order \eref{eq:ordo_p_theta}. Also note that
\begin{equation}
    \|x^*\|_2 = \| \vec{A} \|_2 \leq \| \ones_{X^2} \|_2 = \sqrt{X^2} = X,
\end{equation}
due to the sum-to-one constraint of $A$.

Hence, for every $\varepsilon > 0$ (and $N$ sufficiently large), we have in the
bound \eref{eq:qp_bound}, that
\begin{align}
    \|A - \hat A\|_F &= \|\vec{A} - \vec{\hat A}\|_2 \notag \\
                     &= \|x^* - \hat x^*\|_2 \notag \\
                     &\leq \frac{\delta}{\lambda - \delta} (1 + \|x^*\|_2) \notag \\
                     &\leq \frac{\delta}{\lambda} \; (1 + \|x^*\|_2) \notag \\
                     &\leq \frac{1}{\sqrt{N}} \; \frac{K(\varepsilon) \; (1
+ X)}{\lambda}
    \label{eq:A_hatA_bound}
\end{align}
with probability greater than $1 - \varepsilon$. This shows that\comment{\footnote{Note
    that for $A \in \mathbb{R}^{X\times Y}$, $\|A\|_F \leq M \Rightarrow |a_{ij}| \leq M$ and $|a_{ij}| \leq M \Rightarrow
\|A\|_F \leq \sqrt{XY} M$, so that a stochastic bound on the Frobenius norm is
equivalent to a stochastic elementwise bound.}}
\begin{equation}
    \hat A = A + \mathcal{O}_p(N^{-1/2}).
\end{equation}

\subsubsection{$\sqrt{N}$-consistency of $\hat \pi_\infty$ and $\hat P$}
\label{ssec:error_propagation}

Again, for any $\varepsilon > 0$, we have using equation \eref{eq:recover_pi} that
\begin{align}
    \|\pi_\infty - \hat \pi_\infty \|_2 &= \| A\ones - \hat A \ones \|_2 \notag \\
                                        &= \| (A - \hat A) \ones \|_2 \notag \\
                                        &\leq \Big(\max_{\|y\|_2 = 1} \|(A-\hat A)y\|_2\Big) \|\ones\|_2 \notag \\
                                        &\leq \|A-\hat A\|_F \|\ones\|_2 \notag \\
                                        &= \|A-\hat A\|_F \sqrt{X} \notag \\
                                        &\leq \frac{1}{\sqrt{N}} \; \frac{K(\varepsilon) \; (1
+ X)\sqrt{X}}{\lambda}
\label{eq:pi_hatPi_bound}
\end{align}
holds with probability greater than $1-\varepsilon$.

Continuing, equations \eref{eq:recover_pi} and \eref{eq:recover_P} tell us that
\begin{equation}
    \|P - \hat P\|_F = \|\diag{A \ones}^{-1} A - \diag{\hat A \ones}^{-1} \hat A\|_F.
\end{equation}
We will use that, for two invertible diagonal matrices $D_1$ and $D_2$, and
arbitrary matrices $X$ and $Y$, it holds that
\begin{align}
    \|D_1^{-1}X &- D_2^{-1}Y\|_F = \|D_1^{-1}D_2^{-1} (D_2 X - D_1 Y) \|_F \notag \\
                                &\leq \|D_1^{-1} D_2^{-1}\|_F \| D_2 X - D_1 Y \|_F \notag \\
                                &\leq \|D_1^{-1}\|_F \|D_2^{-1}\|_F \| D_2 X - D_1 Y + D_1 X - D_1 X \|_F \notag \\
                                &= \|D_1^{-1}\|_F \|D_2^{-1}\|_F \| (D_2 - D_1) X + D_1 (X-Y) \|_F \notag \\
                                &\leq \|D_1^{-1}\|_F \|D_2^{-1}\|_F \notag \\
                                &\quad\quad \times \Big( \|
    D_2 - D_1\|_F \|X\|_F + \|D_1\|_F \|X-Y \|_F \Big). 
\end{align}
This yields
\begin{align}
    \|P - \hat P\|_F \leq \|&\diag{A \ones}^{-1}\|_F \|\diag{\hat A \ones}^{-1}\|_F \notag \\
    &\times \Big( \|\hat A \ones - A \ones\|_2 \|A\|_F + \|A\ones\|_2 \|A - \hat A\|_F \Big),
\end{align}
where the first factor is bounded
by a constant due to the ergodicity assumptions (the stationary distribution
has strictly positive elements) and the terms in the parenthesis have trivial
bounds, or can be bounded using equations \eref{eq:A_hatA_bound} and
\eref{eq:pi_hatPi_bound}. Hence, for any
$\varepsilon > 0$, we have that $\|P - \hat P\|_F
\leq \frac{\text{constant}}{\sqrt{N}}$ with probability greater than
$1-\varepsilon$, or equivalently that,
\begin{equation}
    \hat P = P + \mathcal{O}_p(N^{-1/2}).
    \label{eq:Op_P}
\end{equation}

\subsubsection{$\Delta$-method}

Assume that the parametrization of the transition matrix is continuous and differentiable,
and denote by $\theta^*$ the true parameters. We can then propagate relation
\eref{eq:Op_P} to the parameters $\theta$ to obtain
\begin{equation}
    \hat \theta_\text{MM} = \theta^* + \mathcal{O}_p(N^{-1/2}),
\end{equation}
using the $\Delta$-method -- in particular, the first part of the proof of Theorem 3.1 in
\cite{vaart_asymptotic_1998}.  

\subsubsection{Regularity of the likelihood function}
\label{ssec:regularity_loglik}
Denote the Fisher information matrix as $I_F(\theta^*)$. Then Theorems
12.5.5 and 12.5.6 of \cite{cappe_inference_2005} guarantee that we have
a central limit theorem for the score function and a law of large numbers for
the observed information matrix, as follows:
\begin{align}
    N^{-1/2} \nabla_\theta l_N(\theta^*) &\rightarrow_d \mathcal{N}(0,
    I_F(\theta^*), \\
    N^{-1} \nabla_\theta^2 l_N(\theta^*) &\rightarrow_p -I_F(\theta^*),
\end{align}
as $N \rightarrow \infty$, since our chain is finite and $P, B > 0$.

\comment{
We need to verify the assumptions of the theorems:
\begin{itemize}
    \item[12.0.1] (i) and (ii) follows from that we consider finite HMMs
        (existence of transition/observation densities).
        (iii) from the assumption of $P > 0$ and $B > 0$.
    \item[12.2.1] (i) from $P > 0$. (ii) from $B > 0$.
    \item[12.3.1] follows from finite HMM. for $b^+$, $g_\theta \leq 1$, and
        $b^- > 0$ due to $B > 0$, so $\log b^-$ is bounded.
    \item[12.5.1] the parametrization is direct (every element of $P$ is one
        element of $\theta$). (i) are ok by direct, or exponential parameterization
        (composition of continuous functions) -- everything is infinitely
        differentiable. (ii) and (iii) note that $\nabla_\theta \log q_\theta
        \sim \frac{1}{q_\theta} \nabla_\theta q_\theta$ and since $q_\theta
        > 0$ and $\nabla_\theta q_\theta = 1$ (direct parametrization), it is
        ok. Same goes for the second derivative. (iv) all
        elements of $g_\theta(x,y) = B_{x,y}$ are bounded by 1. 
\end{itemize}
}

\subsubsection{Asymptotic efficiency by Newton-Raphson}

\begin{lemma}
    Let $\thetainit = \theta^* + \mathcal{O}_p(N^{-1/2})$. Then, one
    Newton-Raphson step starting in $\thetainit$ on $l_N(\theta)$ gives an
    asymptotically efficient estimator, i.e., with 
    \begin{equation}
        \hat \theta_\text{NR} = \thetainit - \big[\nabla_\theta^2 l_N(
        \thetainit)\big]^{-1} \; \nabla_\theta l_N(\thetainit),
    \end{equation}
    we get
    \begin{equation}
        \sqrt{N}(\hat \theta_\text{NR} - \theta^*) \rightarrow_d
        \mathcal{N}(0, I_F^{-1}(\theta^*)).
    \end{equation}
    \label{lemma:newton_efficient}
\end{lemma}
(Proof on next page)
\clearpage
\begin{proof}
    We have that
    \begin{align}
        \sqrt{N}(\thetaNR - \theta^*) &= \sqrt{N}(\thetainit - \theta^*)
        - \sqrt{N} \big[\hesslik(\thetainit) \big]^{-1} \gradlik(\thetainit)
        \notag \\
        &= \sqrt{N}(\thetainit - \theta^*) - N^{-1/2} \big[ -\Fisher(\theta^*)
    + \op(1) \big]^{-1} \big[ \gradlik(\theta^*)
        + \hesslik(\theta^*)(\thetainit - \theta^*)
    + \op(1) \big] \notag \\
    &= \sqrt{N}(\thetainit - \theta^*) - \big[ - \Fisher(\theta^*)
+ \op(1)\big]^{-1} \Big[ N^{-1/2} \gradlik(\theta^*)
    + \sqrt{N}[-\Fisher(\theta^*) + \op(1)](\thetainit - \theta^*) + \op(1)
\Big] \notag \\
&= \sqrt{N}(\thetainit - \theta^*) + \Fisher^{-1}(\theta^*) \Big[ N^{-1/2}
\gradlik(\theta^*) - \Fisher(\theta^*) \sqrt{N} (\thetainit - \theta^*) \Big]
+ \op(1) \notag \\
&= \Fisher^{-1}(\theta^*) N^{-1/2} \gradlik(\theta^*) + \op(1) \notag \\
&\to_d \mathcal{N}(0, \Fisher^{-1}(\theta^*) \Fisher(\theta^*)
\Fisher^{-T}(\theta^*) \notag \\
&= \mathcal{N}(0, \Fisher^{-1}(\theta^*).
    \end{align}
\end{proof}

\lref{lemma:newton_efficient}, together with the results of
subsections~\ref{ssec:error_propagation} and \ref{ssec:regularity_loglik}, conclude
the proof of \thrmref{thrm:as_eff} by taking $\thetainit = \hat \theta_{MM}$.

\comment{
\section{Random Systems}

Details on how the random systems were generated?
}

\end{document}